\documentclass{article}
\usepackage{Sub/spconf,amsmath,graphicx,hyperref}

\usepackage{balance}
\usepackage{epstopdf}

\usepackage{amsmath,amssymb,epsfig,cite,subfigure}
\include{macros}
\usepackage{graphicx}
\usepackage{epstopdf}

\usepackage{acronym}

\usepackage{bm}

\usepackage{pifont}

\usepackage{tikz}
\usetikzlibrary{calc}

\usepackage{tcolorbox}

\usepackage{accents}

\allowdisplaybreaks

\usepackage{soul}

\usepackage{amsthm}

\usepackage{nicefrac}

\usepackage{lipsum,multicol}

\usepackage{algorithm}
\usepackage{algpseudocode}

\algdef{SE}[SUBALG]{Indent}{EndIndent}{}{\algorithmicend\ }%
\algtext*{Indent}
\algtext*{EndIndent}

\usepackage{enumitem}
\usepackage{mwe}    
\usepackage{epsfig}
\usepackage{subfigure}
\usepackage{color}
\usepackage{url}
\usepackage{mathtools,xparse}

\usepackage{gensymb}

\usepackage[multiple]{footmisc}

\usepackage{xr}

\makeatletter
\newcommand*{\addFileDependency}[1]{
  \typeout{(#1)}
  \@addtofilelist{#1}
  \IfFileExists{#1}{}{\typeout{No file #1.}}
}
\makeatother

\newcommand*{\myexternaldocument}[1]{%
    \externaldocument{#1}%
    \addFileDependency{#1.tex}%
    \addFileDependency{#1.aux}%
}

\myexternaldocument{supp}

\usepackage{xcolor}

\newcommand\abss[1]{\lvert#1\rvert}

\newcommand{\thn}[1]{ {#1^{\rm{th} } } }

\newcommand{\Eee}{\mathbb{E}}

\newcommand{\snr}{{\rm{SNR}}}

\newcommand{\pd}{P_{\rm{d}}}

\newcommand{\snri}{\snr^{-1}}

\newcommand{\sigmax}{\sigma_x}

\newcommand{\sigmaa}{\sigma_{\alpha_k}}
\newcommand{\sigmaaa}{\sigma_{\alpha}}

\newcommand{\xdelta}{x_{\Delta}}
\newcommand{\xhat}{\widehat{x}}

\newcommand{\Tsym}{ T_{\rm{sym}} }

\newcommand{\Tcp}{ T_{\rm{cp}} }

\newcommand{\deltaf}{ \Delta f }

\newcommand{\fc}{ f_c }
\newcommand{\Ntx}{ N_{\rm{T}} }

\newcommand{\atx}{ \aaa_{\rm{T}} }

\newcommand{\HHhatmmse}{ \widehat{\HH}^{\rm{MMSE}} }

\newcommand{\hhatmmse}{\widehat{h}^{\rm{MMSE}}}

\newcommand{\znm}{ z_{n,m} }

\newcommand{\fftx}{ \ff_{\rm{T}} }

\newcommand{\mtCN}{{\mathcal{CN}}}

\newcommand{\vecc}[1]{ {\rm{vec}}\left(#1\right)  }

\newcommand{\Imatrix}{{ \boldsymbol{\mathrm{I}} }}

\newcommand{\aaa}{\mathbf{a}}
\newcommand{\cc}{ \mathbf{c} }
\newcommand{\bb}{ \mathbf{b} }

\newcommand{\boldzero}{{ {\boldsymbol{0}} }}

\newcommand{\Pt}{{ P_{\rm{T}} }}

\newcommand{\sigmarcs}{{ \sigma_{{\rm{rcs}}} }}

\newcommand{\sigmarcsmin}{{ \sigma_{{\rm{rcs, min}}} }}

\newcommand{\ff}{\mathbf{f}}

\newcommand{\YY}{ \mathbf{Y} }

\newcommand{\HH}{ \mathbf{H} }

\newcommand{\XX}{ \mathbf{X} }

\newcommand{\ZZ}{ \mathbf{Z} }

\newcommand{\boldG}{ \mathbf{G} }

\newcommand{\pslrrf}{\pslr^{\rm{RF}}}
\newcommand{\pslrmf}{\pslr^{\rm{MF}}}

\newcommand{\pslr}{{\rm{PSLR}}}

\newcommand{\Rside}{\mathcal{R}}

\newcommand{\pp}{ \mathbf{p} }

\newcommand{\trp}{\mathsf{T}}

\newcommand{\cset}[2]{ \mathbb{C}^{#1 \times #2}  }

\newcommand{\conj}{ {\ast} }

\newcommand{\XXhat}{ \widehat{\mathbf{X}} }
\newcommand{\XXdelta}{ {\mathbf{X}}_{\Delta} }

\makeatletter \renewcommand\d[1]{\ensuremath{%
		\;\mathrm{d}#1\@ifnextchar\d{\!}{}}}
\makeatother

\makeatletter
\newcommand*\rel@kern[1]{\kern#1\dimexpr\macc@kerna}
\newcommand*\widebar[1]{%
  \begingroup
  \def\mathaccent##1##2{%
    \rel@kern{0.8}%
    \overline{\rel@kern{-0.8}\macc@nucleus\rel@kern{0.2}}%
    \rel@kern{-0.2}%
  }%
  \macc@depth\@ne
  \let\math@bgroup\@empty \let\math@egroup\macc@set@skewchar
  \mathsurround\z@ \frozen@everymath{\mathgroup\macc@group\relax}%
  \macc@set@skewchar\relax
  \let\mathaccentV\macc@nested@a
  \macc@nested@a\relax111{#1}%
  \endgroup
}
\makeatother

\theoremstyle{remark}

\newtheoremstyle{mytheoremstyle} 
    {\topsep}                    
    {\topsep}                    
    {\upshape}                   
    {.5em}                           
    {\itshape}                   
    {.}                          
    {.5em}                       
    {}  

\theoremstyle{plain}

\newtheoremstyle{iremark}
  {\topsep}   
  {\topsep}   
  {\upshape}  
  {0.2in}       
  {\itshape}  
  {.}         
  {5pt plus 1pt minus 1pt} 
  {\thmname{#1}\thmnumber{ \itshape#2}\thmnote{ (#3)}} 

\newtheorem{proposition}{Proposition}

\theoremstyle{definition}
\newtheorem*{proof}{Proof}

\theoremstyle{definition}
\newtheorem{remark}{Remark}

\newtheoremstyle{nopunctuation}
  {\topsep}   
  {\topsep}   
  {\itshape}  
  {}          
  {\bfseries} 
  {}          
  { }         
  {\thmname{#1}\ \thmnumber{#2}\ \thmnote{#3}} 

\theoremstyle{italic}
\newtheorem{corollary}{Corollary} 
\acrodef{RIS}{reconfigurable intelligent surface}
\acrodef{RCS}{radar cross section}
\acrodef{SNR}{signal-to-noise ratio}
\acrodef{ISAC}{integrated sensing and communication}
\acrodef{ISLAC}{integrated sensing, localization, and communication}
\acrodef{LOS}{line-of-sight}
\acrodef{NLOS}{non-line-of-sight}
\acrodef{AOA}{angle-of-arrival}
\acrodef{AOD}{angle-of-departure}
\acrodef{UE}{user equipment}
\acrodef{NF}{near-field}
\acrodef{BS}{base station}
\acrodef{MCRB}{misspecified Cram\'{e}r-Rao bound}
\acrodef{CRB}{Cram\'{e}r-Rao bound}
\acrodef{LB}{lower bound}
\acrodef{ML}{maximum-likelihood}
\acrodef{MML}{mismatched maximum-likelihood}
\acrodef{DL}{downlink}
\acrodef{UL}{uplink}
\acrodef{MIMO}{multiple-input multiple-output}
\acrodef{MISO}{multiple-input single-output}
\acrodef{SISO}{single-input single-output}
\acrodef{SIP}{shift invariance property}
\acrodef{FIM}{Fisher information matrix}
\acrodef{RMSE}{root mean-squared error}
\acrodef{AWGN}{additive white Gaussian noise}
\acrodef{ADMM}{alternating direction method of multipliers}
\acrodef{LS}{least-squares}
\acrodef{SOC}{second-order cone}
\acrodef{DoF}{degree-of-freedom}
\acrodef{CFO}{carrier frequency offset}
\acrodef{TO}{timing offset}
\acrodef{ICI}{intercarrier interference}
\acrodef{i.i.d.}{independently and identically distributed}
\acrodef{MI}{mutual information}
\acrodef{SAC}[S\&C]{sensing and communication}

\acrodef{ULA}{uniform linear array}

\acrodef{SCE}{sensing channel estimator}
\acrodef{RF}{reciprocal filtering}
\acrodef{MF}{matched filtering}
\acrodef{LMMSE}{linear minimum mean-squared-error}
\acrodef{MMSE}{minimum mean-squared-error}
\acrodef{BICM}{bit-interleaved coded modulation}

\acrodef{LLR}{log-likelihood ratio}

\acrodef{SSB}{synchronization signal block}
\acrodef{DMRS}{demodulation reference signal}
\acrodef{PTRS}{phase tracking reference signal}

\acrodef{LO}{local oscillator}
\acrodef{PPM}{parts-per-million} 

\acrodef{DFT}{discrete Fourier transform} 
\acrodef{PSD}{power spectral density}

\acrodef{PSLR}{peak-to-sidelobe level ratio}
\acrodef{DPU}{data payload uncertainty}
\acrodef{CP}{cyclic prefix}
\acrodef{PCS}{probabilistic constellation shaping}
\acrodef{CFAR}{constant false alarm rate}

\allowdisplaybreaks

\allowdisplaybreaks

\makeatletter
\renewcommand\section{\@startsection{section}{1}{\z@}%
 {1.6ex plus .3ex minus .2ex}%
 {.75ex plus .2ex}%
 {\normalfont\normalsize\bfseries}}
\renewcommand\subsection{\@startsection{subsection}{2}{\z@}%
 {1.1ex plus .25ex minus .15ex}%
 {.5ex plus .15ex}%
 {\normalfont\normalsize\bfseries}}
\def\abstract{\begin{center}
 {\bf ABSTRACT\vspace{-.15em}\vspace{0pt}}
 \end{center}}

\long\def\@makecaption#1#2{%
 \vskip 3pt%
 \setbox\@tempboxa\hbox{#1. #2}%
 \ifdim \wd\@tempboxa >\hsize #1. #2\par \else
 \hbox to\hsize{\hfil\box\@tempboxa\hfil}%
 \fi}
\def\thebibliography#1{\section{References}\footnotesize\list
 {[\arabic{enumi}]}{\settowidth\labelwidth{[#1]}\leftmargin\labelwidth
 \advance\leftmargin\labelsep
 \usecounter{enumi}
 \setlength{\itemsep}{0pt}
 \setlength{\parsep}{0pt}
 \setlength{\parskip}{0pt}}
 \def\newblock{\hskip .11em plus .33em minus .07em}
 \sloppy\clubpenalty4000\widowpenalty4000
 \sfcode`\.=1000\relax}
\makeatother

\title{Bistatic ISAC under Data Payload Uncertainty}

\name{ \begin{tabular}{c}
     Musa Furkan Keskin$^{\ast}$,  Silvia Mura$^{\star}$, Marouan Mizmizi$^{\star}$,   Dario Tagliaferri$^{\star}$, Fan Liu$^{\dagger}$, Henk Wymeersch$^{\ast}$ 
\end{tabular} \thanks{This work is supported by the Swedish Research Council (VR) through the project 6G-PERCEF under Grant 2024-04390, and by the SNS JU project 6G-DISAC under the EU’s Horizon Europe research and innovation programme under Grant Agreement No 101139130.}}

\address{$^{\ast}$ \small Department of Electrical Engineering, Chalmers University of Technology, Sweden \\ $^{\star}$ \small Department of Electronics, Information and Bioengineering, Politecnico di Milano, Italy \\ \small
$^{\dagger}$ School of Information Science and Engineering, Southeast University, China}

\begin{document}
\ninept
%
\maketitle
\begin{abstract}
We investigate the problem of \ac{DPU} in bistatic integrated sensing and communication (ISAC) systems, which refers to imperfect knowledge of transmit data symbols at the sensing receiver. DPU may result from limited fronthaul capacity in base station (BS)-to-BS bistatic sensing or data demodulation errors in BS-to-user equipment (UE) bistatic sensing, and can impair target detection through elevated sidelobe levels in the delay-Doppler domain. To characterize this effect, we derive closed-form expressions for the \ac{PSLR} of two common OFDM receivers, namely \ac{RF} and \ac{MF}, under \ac{DPU}. The analysis shows that the \ac{DPU}-induced sidelobe degradation of \ac{RF} depends on the inverse second-, fourth-, and sixth-order constellation moments, whereas that of \ac{MF} depends only on the direct fourth-order moment, making \ac{RF} considerably more vulnerable to \ac{DPU}. To mitigate the resulting sidelobes, we further propose a novel, \ac{DPU}-aware \ac{MMSE} receiver. 
Simulation results validate the \ac{PSLR} analysis and demonstrate superior detection performance of \ac{MMSE} over \ac{RF} and \ac{MF} while revealing the extent to which it can recover \ac{DPU}-induced performance loss relative to a genie-aided receiver with perfect knowledge of data payload.
\end{abstract}
\acresetall
\section{Introduction}
\label{sec_intro}

Integrated sensing and communication (ISAC) enables simultaneous data transmission and environmental sensing, and is envisioned as a key component of 6G wireless networks \cite{Fan_ISAC_6G_JSAC_2022}. While \textit{monostatic ISAC}, with co-located transmitter and receiver, has been extensively studied \cite{5G_NR_JRC_analysis_JSAC_2022,ISAC_OFDM_Pilots,Liu_Reshaping_OFDM_ISAC_2023}, \textit{bistatic ISAC} with spatially separated transmit and receive nodes has recently gained increasing attention \cite{bistatic_OFDM_ISAC_OTA_2024,bistaticISAC_Clock_Async_2024,bistatic_ISAC_joint_2024,bistatic_data_payload_ICC_2026}. Bistatic operation avoids full-duplex transceiver complexity while extending sensing coverage and providing spatial diversity for target detection \cite{wavefieldSensing_2025,bistaticISAC_Clock_Async_2024}.

Sensing in communication systems has conventionally relied on deterministic pilot signals. However, their limited time-frequency occupancy restricts the sensing SNR and their sparse placement can introduce grating lobes or reduce delay-Doppler resolution \cite{Fan_CommStandMag_2026_Payload,ISAC_OFDM_Pilots}. Payload-based sensing overcomes these limitations by exploiting communication data in addition to pilots and has shown substantial potential in monostatic ISAC, where the transmitted symbols are readily available at the co-located sensing receiver \cite{sensWithComm_Fan_JSAC_2026,holisticISAC_2024,Kaitao_ISAC_TSP_2026}. Extending this approach to bistatic sensing is less straightforward since the spatially separated receiver does not have access to the transmitted data payload. Existing bistatic approaches therefore exploit the prior distribution of the constellation instead of the symbols themselves, jointly processing deterministic pilots and unknown random payloads to improve detection over pilot-only sensing \cite{fan_bistatic_pilot_random_data_2025,bistatic_data_payload_ICC_2026}.

Between perfectly known data in monostatic ISAC and entirely unknown data in bistatic ISAC lies a practically important regime where the sensing receiver has only \textit{imperfect} knowledge of the transmitted symbols, hereafter referred to as \ac{DPU}. Such uncertainty can arise in two representative bistatic scenarios, illustrated in Fig.~\ref{fig_scenario}. In BS-to-BS sensing, limited fronthaul capacity may require compression and quantization of the data conveyed from the transmitting BS to remote sensing nodes, which inevitably introduces errors in data symbols \cite{limited_fronthaul_2025,eckrich_fronthaul_radar_2024,zhang_fronthaul_beamforming_ISAC_2024}. In BS-to-UE sensing, data-aided iterative channel estimation relies on demodulated payload symbols that may contain decision errors \cite{bistatic_ISAC_sync_2023,bistatic_OFDM_ISAC_OTA_2024,bistaticISAC_conference_2025,bistatic_sensing_5G_NR_2025}. The detrimental effect of such uncertainty has already been observed in OFDM-based passive radar studies, where demodulation errors in the reference signal have been shown to increase sidelobe levels and impair target detection \cite{OFDM_Passive_Res_2017_TSP}.

\begin{figure}[t]
	\centering
	\includegraphics[width=0.8\linewidth]{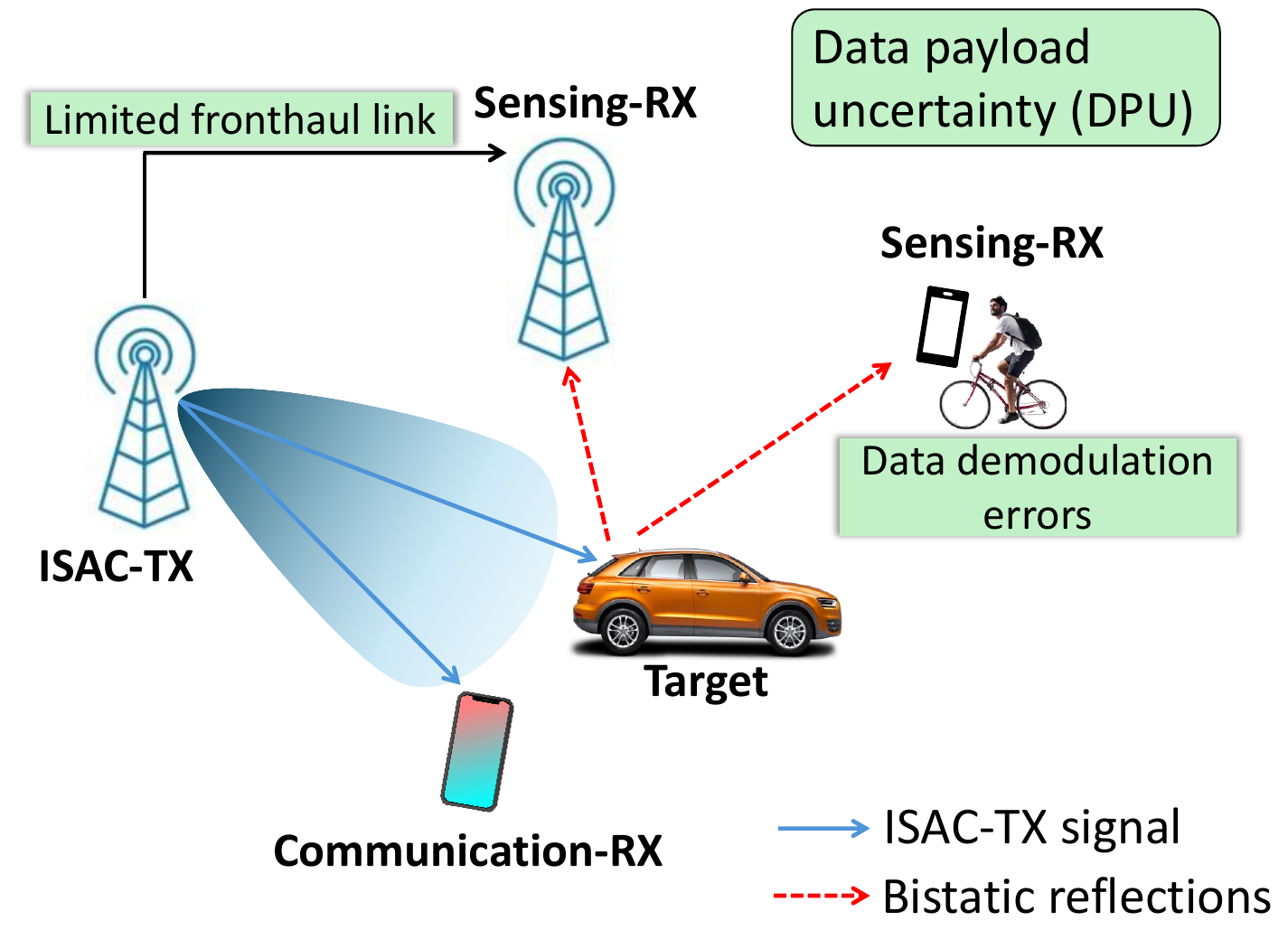}
	\caption{\footnotesize Bistatic ISAC scenario where sensing receivers perform target detection under \ac{DPU}, resulting from limited fronthaul capacity (BS-to-BS bistatic sensing) or data demodulation errors (BS-to-UE bistatic sensing).} 
	\label{fig_scenario}
\end{figure}

Despite its practical relevance, the effect of \ac{DPU} on the sidelobe behavior of OFDM \acp{SCE} remains largely unexplored. Existing analyses of common \acp{SCE}, including \ac{RF} and \ac{MF}, generally assume perfect knowledge of the transmit symbols \cite{holisticISAC_2024,OFDM_Radar_Corr_TAES_2020,Fan_PCS_TWC_2026,WF_LMMSE_TCOM_2023,MIMO_OFDM_ISAC_Waveform_2025_Lee}. The study in \cite{Zhang_mismatched_bistatic_SPAWC_2026} investigates mismatched bistatic OFDM ranging when a hard demodulated symbol sequence is used as the reference for \ac{MF}, and derives the expected delay sidelobe levels induced by SNR-dependent decision errors for BPSK, QPSK and square QAM. Moreover, \cite{bistaticISAC_conference_2025} empirically observes that \ac{RF} can underperform \ac{MF} in the presence of soft demodulation errors. However, these works do not jointly characterize the sidelobe behavior of \ac{RF} and \ac{MF}, and thus cannot establish how the RF-MF ordering depends on the \ac{DPU} level, SNR and data constellation. As a result, it remains unclear how an additive error \ac{DPU} model, applicable to both fronthaul quantization \cite{eckrich_fronthaul_radar_2024,zhang_fronthaul_beamforming_ISAC_2024,limited_fronthaul_2025} and soft demodulation residuals \cite{bistatic_ISAC_sync_2023,bistatic_OFDM_ISAC_OTA_2024,bistatic_sensing_5G_NR_2025}, affects the delay-Doppler sidelobe performance of both \ac{RF} and \ac{MF} for generic constellations, which estimator is inherently more robust to such uncertainty, and whether knowledge of the symbol error statistics can be exploited to suppress the resulting sidelobes.

To address this gap, we characterize and mitigate the impact of \ac{DPU} on sensing in OFDM-based bistatic ISAC systems. Our contributions are as follows: \textit{(i)} We derive closed-form \ac{PSLR} expressions for \ac{RF} and \ac{MF} under \ac{DPU} for arbitrary constellations and SNRs. The analysis reveals that the \ac{DPU}-induced sidelobe degradation of \ac{RF} depends on the inverse second-, fourth-, and sixth-order constellation moments, whereas the \ac{MF} result is governed by the direct fourth-order moment, revealing the substantially higher sensitivity of \ac{RF} to \ac{DPU}; \textit{(ii)} we establish analytical RF-MF ordering and monotonicity results with respect to SNR and DPU level; and   \textit{(iii)} we develop a \ac{DPU}-aware \ac{MMSE} channel estimator that jointly exploits the constellation prior and symbol error statistics. Numerical results validate the \ac{PSLR} analysis and demonstrate improved performance of \ac{MMSE} over \ac{RF} and \ac{MF}, closing the gap to a genie-aided estimator with perfect knowledge of the data payload.

\section{System Model}

We consider a bistatic OFDM ISAC system consisting of three entities, namely \textit{(i)} a multiple-antenna ISAC transmitter (ISAC-TX), \textit{(ii)} a single-antenna bistatic sensing receiver (S-RX) and \textit{(iii)} a single-antenna communication receiver (C-RX), as well as a number of targets in the environment, as shown in Fig.~\ref{fig_scenario}. ISAC-TX sends data symbols to C-RX for downlink communications, while S-RX simultaneously performs bistatic sensing using the received target echoes \cite{info_theory_bistatic_ISAC_2025_TIT,bistatic_data_aided_MIMO_OFDM_ISAC_2025}. In the bistatic scenario under consideration, the goal of S-RX is to detect the targets in the environment and estimate their delay and Doppler parameters under \ac{DPU}, i.e., using noisy estimates of data symbols. Meanwhile, C-RX carries out standard OFDM receive processing tasks \cite{demo_bistatic_ISAC,bistaticISAC_Clock_Async_2024}. 

\subsection{Transmit Signal Model}
ISAC-TX transmits an OFDM waveform with $N$ subcarriers and $M$ symbols. The total duration of each symbol is $\Tsym = \Tcp + T$, where $\Tcp$, $T$ and $\deltaf = 1/T$ denotes the \ac{CP} duration, the useful symbol duration and the subcarrier spacing, respectively \cite{RadCom_Proc_IEEE_2011}. Let $\XX \in \cset{N}{M}$ denote the transmit symbol matrix. The entries of $\XX$ are i.i.d. complex random variables represented by $x$, drawn from a given discrete constellation with unit average power, i.e., $\Eee\{ \abss{x}^2 \} = 1$.

\subsection{Signal Model at Bistatic Sensing Receiver}
In the bistatic sensing channel between ISAC-TX and S-RX, we consider the presence of $K$ targets, along with a \ac{LOS} path indexed by $k=0$. The $\thn{k}$ path is characterized by a complex channel gain $\alpha_k$ (including the effects of bistatic \ac{RCS} for $k>0$, TX beamforming gain $\atx^\trp(\theta_k) \fftx$ and path attenuation), a delay $\tau_k$, a Doppler shift $\nu_k$ and an \ac{AOD} $\theta_k$. Here, $\atx(\theta) \in \cset{\Ntx}{1}$ is the steering vector of the \ac{ULA} at ISAC-TX, $\fftx \in \cset{\Ntx}{1}$ is the TX beamforming vector and $\Ntx$ is the number of antenna elements at the ISAC-TX array.

After removing the \ac{CP} and performing an $N$-point \ac{DFT} over fast-time samples for each OFDM symbol, the frequency-domain observation matrix at S-RX over $M$ symbols is given by 
\begin{align} \label{eq_YY}
    \YY &= \XX \odot \HH + \ZZ \in \cset{N}{M} \,,
\end{align}
where $\odot$ denotes Hadamard (element-wise) multiplication, $\ZZ \in \cset{N}{M}$ is \ac{AWGN} matrix with i.i.d entries $\znm \triangleq [\ZZ]_{n,m} \sim \mtCN(0, \sigma^2) $, $\sigma^2 = N_0 N \deltaf N_F$ with $N_0$ and $N_F$ representing the noise \ac{PSD} and the noise figure, respectively, and
\vspace{-0.1in}
\begin{align} \label{eq_hh}
    \HH &= \sum_{k=0}^{K}   \alpha_{k} \,  \bb(\tau_{k}) \cc^\trp(\nu_{k}) \in \cset{N}{M}  \,
    \vspace{-0.5in}
\end{align}
denotes the bistatic sensing channel\footnote{Here, $\tau_k$ and $\nu_k$ denote the delay and Doppler of the $\thn{k}$ path biased by \ac{TO} and \ac{CFO}, respectively, due to TX-RX asynchronism \cite{bistaticISAC_Clock_Async_2024}. We focus on their estimation, leaving synchronization and target localization to subsequent processing.} in the frequency/slow-time domain \cite{LiuSeventy23,SPM_ISAC_multicarrier_2024}. In \eqref{eq_hh}, $\bb(\tau) \in \cset{N}{1}$ with $[\bb(\tau)]_n = e^{-j 2 \pi n \deltaf  \tau}$ and $\cc(\nu) \in \cset{M}{1}$ with  $[\cc(\nu)]_m = e^{j 2 \pi m \Tsym \nu }$ represent, respectively, the frequency-domain and slow-time-domain steering vectors. In obtaining \eqref{eq_YY}, we make the standard assumption $\Tcp \geq \max_k \tau_k - \min_k \tau_k$, i.e., the delay spread of the sensing channel is smaller than \ac{CP} \cite{SPM_JRC_2019,SPM_ISAC_multicarrier_2024}.

\begin{figure}[t]
	\centering
	\includegraphics[width=0.9\linewidth]{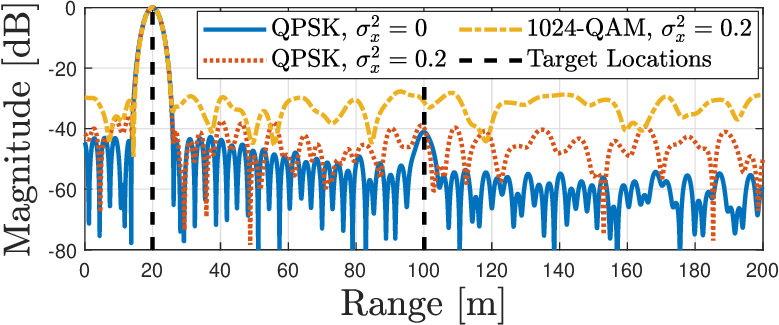}
	\vspace{-0.05in}
	\caption{\footnotesize  Range profiles obtained via conventional reciprocal filtering (RF) based OFDM channel estimation \cite{OFDM_Radar_Corr_TAES_2020,holisticISAC_2024}, with and without \ac{DPU} for QPSK and $1024$-QAM constellations, using $N=400$ and the remaining parameters specified in Sec.~\ref{sec_num}. The scenario contains two targets with locations $[56.9, 10] \, \rm{m}$ and $[81.2, 46.9] \, \rm{m}$, velocities $[1.4, -2.2] \, \rm{m/s}$ and $[3, -3.2] \, \rm{m/s}$, and RCSs $25 \, \rm{dBsm}$ and $15 \, \rm{dBsm}$. The range values represent the differential ranges with respect to the \ac{LOS} path.} 
	\label{fig_range_profile_DPU}
\end{figure}

\subsection{Goal of Bistatic Sensing Receiver}\label{sec_goal_sensing}
At S-RX, the transmit data symbols $\XX$ are \textit{imperfectly known} due to \ac{DPU}, which may result from, for example, limited fronthaul capacity \cite{eckrich_fronthaul_radar_2024,zhang_fronthaul_beamforming_ISAC_2024,limited_fronthaul_2025} or demodulation errors in data-aided sensing \cite{OFDM_Passive_Res_2017_TSP,bistatic_OFDM_ISAC_OTA_2024,bistaticISAC_conference_2025,bistatic_sensing_5G_NR_2025}. To model the impact of \ac{DPU}, we assume that a noisy estimate $\XXhat$ of $\XX$ is available at S-RX such that
\begin{align}\label{eq_xxhat_delta}
    \XXhat = \XX + \XXdelta \,,
\end{align}
where $\XXdelta \in \cset{N}{M}$ contains the errors in data symbols, with $\vecc{\XXdelta} \sim \mtCN(\boldzero, \sigmax^2  \Imatrix)$, and $\sigmax^2$ quantifies the level of \ac{DPU}\footnote{From the limited fronthaul perspective, the error variance $\sigmax^2$ in \eqref{eq_xxhat_delta} can be approximated by the scalar complex Gaussian rate-distortion function as $\sigmax^2 = 2^{-R_f/(NM)}$ \cite[Ch.~10]{cover_thomas_2006}, where $R_f$ is the total fronthaul budget in bits for the entire OFDM frame. From the perspective of data demodulation errors, the model is grounded in the \ac{MMSE} soft symbol estimate framework of \cite{tuchler_mmse_apriori_2002}. In particular, the estimation residual from the \ac{LLR} outputs of the decoder is Gaussian under the standard Gaussian consistency assumption of turbo processing. In both cases, the i.i.d. structure of $\XXdelta$ follows from random interleaving in \ac{BICM} \cite{caire_bicm_1998}.}. We denote the entries of $\XXdelta$, which are i.i.d. complex random variables, by $\xdelta$, which is assumed to be independent of $x$.

Given the observations $\YY$ in \eqref{eq_YY} and the noisy data symbols $\XXhat$ in \eqref{eq_xxhat_delta} at S-RX, the goal is to detect the presence of multiple targets in the sensing channel $\HH$ in \eqref{eq_hh} and estimate their delay-Doppler parameters. To illustrate the impact of \ac{DPU}, we plot in Fig.~\ref{fig_range_profile_DPU} the range profiles at S-RX using different data constellations, which reveal the target masking effect of \ac{DPU}, leading to degradation in detection performance. To quantify \ac{DPU}-induced performance loss, we next analyze sidelobe levels achieved by common OFDM \acp{SCE}.

\section{Sidelobe Level Analysis under Data Payload Uncertainty}\label{sec_pslr}
In this section, we perform \ac{PSLR} analysis of RF and MF, two common OFDM \acp{SCE}, to provide theoretical characterization of delay-Doppler ambiguity sidelobe levels at S-RX under \ac{DPU}.
We consider the delay-Doppler response $\chi(\tau, \nu)$ of a target with channel gain $\alpha$, delay $\tau_0$ and Doppler shift $\nu_0$, yielding the sensing channel $\HH = \alpha \, \bb(\tau_0) \cc^\trp(\nu_0) \in \cset{N}{M}$ according to \eqref{eq_hh}.
Then, the PSLR is given by \cite{holisticISAC_2024}
\begin{align}\label{eq_pslr_main}
    \pslr = \frac{ \Eee\{ \abss{ \chi(\tau_0, \nu_0) }^2 \}  }{ \Eee\{ \abss{ \chi(\tau, \nu) }^2 \} } \,, \, (\tau, \nu) \in \Rside \,,
\end{align}
where $\Rside = \{(\tau,\nu)~\lvert~\abss{\tau-\tau_0}\gg\Delta\tau,\ \abss{\nu-\nu_0}\gg\Delta\nu\}$ defines the asymptotic side-lobe region, with $\Delta\tau$ and $\Delta\nu$ denoting the delay and Doppler resolutions, respectively. In \eqref{eq_pslr_main}, the delay-Doppler response corresponding to the observation $\YY$ in \eqref{eq_YY} is given by
\begin{align}\label{eq_2d_fft_res_main}
     \chi(\tau, \nu) = \frac{1}{NM} \bb^H(\tau) ( \YY \odot \boldG ) \cc^\conj(\nu) \,,
\end{align}
where $\YY \odot \boldG$ represents the sensing channel estimate obtained using RF or MF. Here, $\boldG \in \cset{N}{M}$ denotes the equalizer/compensator matrix given by 
\begin{align} \label{eq_G_main}
    [\boldG]_{n,m} = \begin{cases}
        1/[\XXhat]_{n,m} \,, & {\rm{RF}}
        \\
        [\XXhat]_{n,m}^\conj \,, & {\rm{MF}}
    \end{cases} \,,
\end{align}
where S-RX relies on noisy estimates of data symbols, $\XXhat$ in \eqref{eq_xxhat_delta}, for equalization as the true data symbols $\XX$ are not available. The following proposition provides the main \ac{PSLR} result.


\begin{proposition}\label{prop_pslr_all}
    The \ac{PSLR} in \eqref{eq_pslr_main} achieved by the considered estimators under \ac{DPU} is given by
\begin{align} \nonumber
    \pslrrf &= 
1 + NM \Big[    
    \snri \mu_{-2} \\ \nonumber
    + &\underbrace{ \sigmax^2 \Big( \mu_{-2} + \snri \mu_{-4} + 2\sigmax^2 \big( \mu_{-4} + \snri \mu_{-6} \big)  \Big)}_{{\rm{sidelobe~increase~due~to~data~payload~uncertainty}}}  
    \Big]^{-1} 
    \\ \nonumber
    \pslrmf &= 1 + NM \Big[   \snri  +  \mu_{4}-1 + \underbrace{\sigmax^2 \big( 1 + \snri \big)}_{\substack{\rm{sidelobe~increase~due~to} \\ \rm{data~payload~uncertainty}}}    \Big]^{-1}    
\end{align}
where $\mu_{\ell} \triangleq \Eee\{\abss{x}^{\ell}\}$, $\snr \triangleq \abss{\alpha}^2/\sigma^2$, $x$ represents the transmit data symbol modeled as a complex random variable from a given constellation with unit average power, i.e., $\mu_2 = 1$, and $\sigmax^2$ denotes the variance of error in estimates of data symbols (i.e., the level of \ac{DPU}), assumed to satisfy $\abss{\xdelta} \ll \abss{x}$ in \eqref{eq_xxhat_delta}.
\end{proposition}
\begin{proof}
    The proof follows by evaluating the expectations in \eqref{eq_pslr_main} by inserting \eqref{eq_YY}, \eqref{eq_xxhat_delta}, \eqref{eq_2d_fft_res_main} and \eqref{eq_G_main}, and applying a $\thn{4}$ order Taylor expansion of their arguments around $\xdelta = 0$, assuming $\abss{\xdelta} \ll \abss{x}$. 
\end{proof}

\begin{remark}[Vulnerability of RF to \ac{DPU}]\label{remark_vuln_rf}
    Prop.~\ref{prop_pslr_all} reveals the significantly higher vulnerability of RF to DPU than MF. The \ac{DPU}-induced sidelobe increase for RF scales with \textit{inverse} second-, fourth- and sixth-order constellation moments, whereas for MF it depends only on the SNR, with its overall \ac{PSLR} governed by the \textit{direct} fourth-order moment. Table~\ref{tab_constellation_moments} shows these moments for different constellations. It is seen that $\mu_4$ increases only mildly with modulation order $M$, whereas the inverse moments grow very rapidly for higher-order $M$-QAM due to constellation points increasingly close to the origin. Hence, the \ac{PSLR} degradation of RF with increasing $M$ becomes more severe as the \ac{DPU} level $\sigmax^2$ grows. In contrast, MF exhibits a milder degradation with $M$ that is independent of $\sigmax^2$, making it much more robust to \ac{DPU}.
\end{remark}

\begin{table}[t]
\centering
\caption{Constellation moments for unit average power.}
\label{tab_constellation_moments}
\footnotesize
\setlength{\tabcolsep}{2.5pt}
\begin{tabular}{c|cccc}
\hline
Modulation & $\mu_4$ & $\mu_{-2}$ & $\mu_{-4}$ & $\mu_{-6}$ \\
\hline
PSK       & 1      & 1      & 1        & 1 \\
16-QAM    & 1.3200 & 1.8889 & 6.8272   & 31.7929 \\
64-QAM    & 1.3810 & 2.6854 & 30.8743  & 589.8164 \\
256-QAM   & 1.3953 & 3.4371 & 127.2873 & 9779.2608 \\
1024-QAM  & 1.3988 & 4.1716 & 512.9992 & 157854.0289 \\
\hline
\end{tabular}
\end{table}

\begin{corollary}\label{prop_mf_monotone}
    The difference between the sidelobe floors (i.e., the sea levels \cite{sensWithComm_Fan_JSAC_2026}) of RF and MF, defined as $(\pslrrf-1)^{-1} - (\pslrmf-1)^{-1}$, is monotonically increasing with respect to $\sigmax^2$ and decreasing with respect to SNR. 
\end{corollary}
\begin{proof}
  The proof follows from $\mu_{-4} \geq \mu_{-2}^2 \geq 1 $ based on Jensen's and Cauchy-Schwarz inequalities.
\end{proof}
\begin{remark}[Robustness of MF]
    Cor.~\ref{prop_mf_monotone} implies that lowering SNR and increasing \ac{DPU} level always favor MF over RF in their relative comparison, which indicates the robustness of MF under severe conditions with regard to noise and \ac{DPU}. 
\end{remark}

\begin{corollary}\label{cor_rf_mf_comparison}
    With PSK constellations, $\pslrmf > \pslrrf$ always holds under \ac{DPU}, i.e., $\sigmax^2 > 0$.
\end{corollary}
\begin{proof}
The proof follows by inserting $\abss{x} = 1$ in Prop.~\ref{prop_pslr_all}. 
\end{proof}
\begin{remark}[Superiority of MF over RF under PSK]
    Unlike the \ac{DPU}-free, monostatic result showing the equivalence of RF and MF with PSK constellations \cite{holisticISAC_2024}, Cor.~\ref{cor_rf_mf_comparison} establishes the superiority of MF over RF when employing PSK under \ac{DPU} in bistatic scenarios. 
\end{remark}

\section{DPU-Aware Estimator Design via MMSE}\label{sec_mmse}
The common \acp{SCE} considered in Sec.~\ref{sec_pslr} do not account for the presence of \ac{DPU}. In this section, we propose an \ac{MMSE} based \ac{SCE} design that takes into account \ac{DPU} and the a-priori statistics of data symbols to improve the resulting sidelobe performance.

To derive the \ac{MMSE} estimate $\HHhatmmse \in \cset{N}{M}$ of the sensing channel $\HH$ in \eqref{eq_hh}, we assume $\alpha_{k} \sim \mtCN(0, \sigmaa^2), \tau_k \sim \mathcal{U}[0, 1/\deltaf], \nu_k \sim \mathcal{U}[0, 1/ \Tsym] $, with the different targets and parameters distributed independently \cite[Lemma~1]{holisticISAC_2024}. Given these statistics of $\HH$, the observations $\YY$ in \eqref{eq_YY} and the erroneous data symbols $\XXhat$ in \eqref{eq_xxhat_delta} at S-RX,
where each entry of $\XX$ is i.i.d. transmit symbol drawn from a finite alphabet $\mathcal{X}=\{x_1,x_2,\ldots,x_S\}$ with the discrete input distribution $\pp = [p_1, \ldots, p_S]^\trp$, each entry of $\HHhatmmse$ is given by \cite[Ch.~11]{kay1993fundamentals}
\begin{align} \label{eq_h_mmse}
    \hhatmmse = \sum_{s=1}^{S} w_{s} \frac{ x_{s}^\conj y}{\abss{x_{s}}^2 + \snri } \,,
\end{align}
where $\snr \triangleq \sigmaaa^2/\sigma^2$, $\sigmaaa^2 \triangleq \sum_{k=0}^K\sigmaa^2$ and
\begin{align} \label{eq_well}
    w_{s} &\triangleq p( x = x_{s} \, \lvert \, y, \xhat) \,, \\ \nonumber
    &= \frac{p_{s} \, \mtCN(\xhat; x_{s}, \sigmax^2)  \mtCN(y; 0, \abss{x_{s}}^2 \sigmaaa^2 + \sigma^2)   }{ \sum_{s=1}^{S} p_{s} \,  \mtCN(\xhat; x_{s}, \sigmax^2)  \mtCN(y; 0, \abss{x_{s}}^2 \sigmaaa^2 + \sigma^2)   } \,.
\end{align}
In \eqref{eq_h_mmse} and \eqref{eq_well}, $\xhat$ and $y$ denote the corresponding entries of $\XXhat$ and $\YY$. To obtain the element-wise MMSE estimator in \eqref{eq_well}, we adopt the Gaussian model $\vecc{\HH} \sim \mtCN(\boldzero, \sigmaaa^2 \Imatrix)$, whose mean and covariance are exact under the stated priors on channel parameters.

\section{Numerical Results}\label{sec_num}
In this section, we perform numerical investigation of \acp{PSLR} derived in Sec.~\ref{sec_pslr} under various operation regimes and assess the sensing performance of the proposed MMSE estimator in Sec.~\ref{sec_mmse}. The default simulation parameters are listed in Table~\ref{tab_SimParam}. We assume the presence of $K=2$ targets in \eqref{eq_hh} along with the LOS path\footnote{We only focus on the target-reflected, \ac{NLOS} paths in \eqref{eq_hh}, assuming the \ac{LOS} path can be estimated and subtracted from \eqref{eq_YY} \cite{data_aided_multistatic_ISAC_2026}.}, where the channel gains are generated according to the bistatic radar equation \cite[Eq.~(2.9)]{Scheer2010RadarEquation}. The locations of ISAC-TX and S-RX are $[0, 0] \, \rm{m}$ and $[50, 0] \, \rm{m}$, while the targets are  located at $[50.3, 8.9] \, \rm{m}$ and $[119.4, 32] \, \rm{m}$ with the respective velocities $[6, -5.2] \, \rm{m/s}$ and $[20, -55.3] \, \rm{m/s}$ and the respective \acp{RCS} $28 \, \rm{dBsm}$ and $\sigmarcs$. To evaluate detection performance, we sweep $\sigmarcs$, employ a \ac{CFAR} detector \cite[Ch.~6.2.4]{richards2005fundamentals} on the range-Doppler plane in \eqref{eq_2d_fft_res_main}, with the probability of false alarm $10^{-4}$, and consider $500$ independent Monte Carlo realizations in \eqref{eq_YY} and \eqref{eq_xxhat_delta}, where each realization corresponds to an \ac{AWGN} matrix $\ZZ$ and a \ac{DPU} matrix $\XXdelta$. We compare the following schemes: \textit{(i) Genie:} the LMMSE estimator with perfect knowledge of $\XX$ (DPU-free baseline with $\sigmax^2 = 0$) \cite{holisticISAC_2024}, \textit{(ii) MMSE:} the proposed DPU-aware estimator in Sec.~\ref{sec_mmse}, \textit{(iii) RF} and \textit{(iv) MF} in \eqref{eq_G_main}.




\begin{table}[t]\footnotesize
\caption{Simulation Parameters}
\centering
\begin{tabular}{ll}
\hline
Carrier frequency & $\fc = 28$ GHz\\
Subcarriers / spacing & $N=1668$, $\deltaf=120$ kHz\\
Useful / CP duration & $T=8.33~\mu$s, $\Tcp=1.16~\mu$s\\
Symbols, transmit power & $M = 60$, $\Pt = 30$ dBm\\
ULA size at ISAC-TX & $\Ntx = 8$ \\
Noise figure / \ac{PSD} & $N_F = 8$ dB, $N_0 = -174$ dBm/Hz \\
\hline
\end{tabular}
\label{tab_SimParam}
\end{table}



\begin{figure}
        \begin{center}
        \subfigure[]{
			 \label{fig_DPU}
			 \includegraphics[width=0.85\columnwidth]{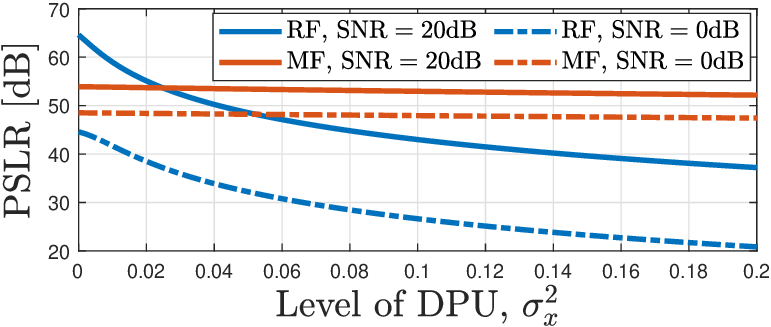}  
		}
        \hspace*{0.02\columnwidth}
        \subfigure[]{
			 \label{fig_mod_order}
			 \includegraphics[width=0.9\columnwidth]{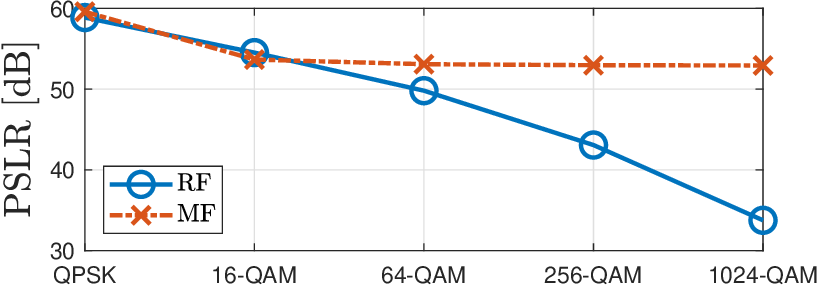}
		}		
		\end{center}
		\vspace{-0.1in}
        \caption{\footnotesize \subref{fig_DPU} PSLR against $\sigmax^2$ for $256$-QAM transmission under different SNRs, and \subref{fig_mod_order} PSLR against modulation order for $\sigmax^2 = 0.1$ and $\snr = 20 \, \rm{dB}$.}  
        \label{fig_pslr_dpu_mod_order}
\end{figure}

\subsection{PSLR Analysis}
To compare the sidelobe performance of \ac{RF} and \ac{MF} under \ac{DPU}, Fig.~\ref{fig_DPU} shows the \acp{PSLR} in Prop.~\ref{prop_pslr_all} as a function of the \ac{DPU} level $\sigmax^2$ for $256$-QAM transmission at two SNRs. The \ac{PSLR} of \ac{RF} degrades sharply with increasing $\sigmax^2$ and decreasing SNR, whereas \ac{MF} remains almost insensitive to $\sigmax^2$ and exhibits a much smaller degradation as SNR decreases. These results demonstrate that \ac{MF} is significantly more robust than \ac{RF} to both increasing \ac{DPU} and decreasing SNR, in agreement with Cor.~\ref{prop_mf_monotone}. Moreover, Fig.~\ref{fig_mod_order} compares the \acp{PSLR} across different constellations, showing that RF degrades far faster than MF with increasing modulation order, in alignment with Remark~\ref{remark_vuln_rf}, Table~\ref{tab_constellation_moments} and the illustrative example in Fig.~\ref{fig_range_profile_DPU}.

\subsection{Detection Performance}
To investigate how the \ac{PSLR} behavior translates into detection performance, Fig.~\ref{fig_RCS} shows the probability of detection ($\pd$) of the second target as a function of its \ac{RCS} $\sigmarcs$ for $\sigmax^2 = 0.1$ and QPSK transmission. The proposed \ac{DPU}-aware \ac{MMSE} estimator significantly outperforms the DPU-ignorant benchmarks, \ac{RF} and \ac{MF}, which demonstrates its effectiveness in suppressing sidelobes and mitigating target masking. This improvement is enabled by exploiting prior statistics of the data symbols and \ac{DPU}, at the cost of increased computational complexity in \eqref{eq_h_mmse} relative to \eqref{eq_G_main}. The \ac{MMSE} estimator recovers a substantial portion of the DPU-induced performance loss, approaching the performance of the genie-aided estimator. Moreover, in agreement with Cor.~\ref{cor_rf_mf_comparison} and Fig.~\ref{fig_mod_order}, \ac{MF} achieves a higher $\pd$ than \ac{RF}.

\begin{figure}
        \begin{center}
        \subfigure[]{
			 \label{fig_RCS}
			 \includegraphics[width=0.85\columnwidth]{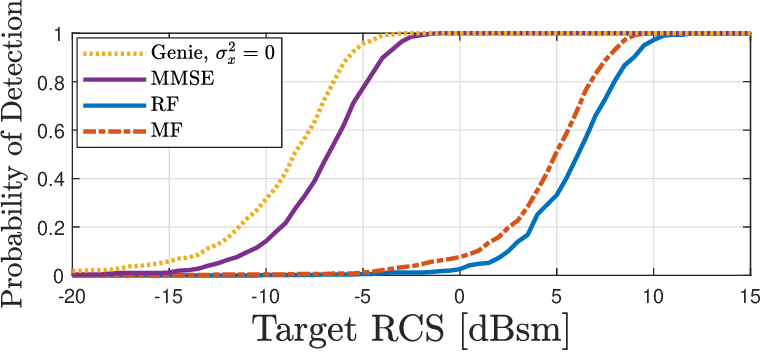}  
		}
        \hspace*{0.02\columnwidth}
        \subfigure[]{
			 \label{fig_min_det_RCS}
			 \includegraphics[width=0.9\columnwidth]{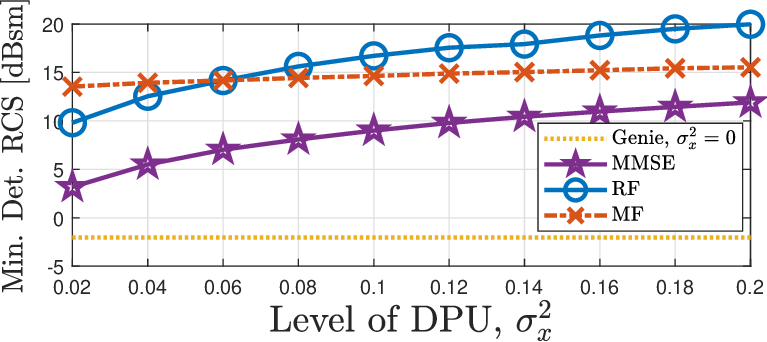}
		}		
		\end{center}
		\vspace{-0.1in}
        \caption{\footnotesize \subref{fig_RCS} Probability of detection with respect to the RCS of the second target for $\sigmax^2 = 0.1$ and QPSK transmission, and \subref{fig_min_det_RCS} Minimum detectable RCS $\sigmarcsmin(\pd^{\star}, \sigmax^2)$ for $\pd^{\star} = 0.9$ with $256$-QAM transmission.}  
        \label{fig_det_perf}
        \vspace{0.05in}
\end{figure}



We finally explore the impact of varying $\sigmax^2$ on detection performance by plotting the minimum detectable \ac{RCS} for a given $\pd^{\star}$, defined as $\sigmarcsmin(\pd^{\star}, \sigmax^2)
\triangleq
\min\left\{
\sigmarcs :
\pd(\sigmarcs, \sigmax^2) \ge \pd^{\star}
\right\}$ \cite{minDetRCS_2019}. Fig.~\ref{fig_min_det_RCS} reports $\sigmarcsmin(\pd^{\star}, \sigmax^2)$ for $\pd^{\star} = 0.9$ as a function of $\sigmax^2$ with $256$-QAM transmission. Consistent with Remark~\ref{remark_vuln_rf}, Cor.~\ref{prop_mf_monotone} and Fig.~\ref{fig_DPU}, \ac{RF} is much more vulnerable than \ac{MF} to increasing \ac{DPU}, leading to an \ac{RCS} loss of $10 \, \rm{dBsm}$ and $2 \, \rm{dBsm}$, respectively, as $\sigmax^2$ increases by a factor of $10$. The \ac{MMSE} estimator maintains the best overall detection performance although it degrades more rapidly with increasing \ac{DPU} than \ac{MF}.

\section{Conclusion}
We have investigated \ac{DPU} in bistatic ISAC systems by performing a \ac{PSLR} analysis of RF and MF and designing a new \ac{DPU}-aware \ac{SCE} based on \ac{MMSE}. The results highlight the significant potential of MF and MMSE under severe \ac{DPU} conditions, indicating their effectiveness in payload-based sensing in practical distributed ISAC systems with limited fronthaul links. Future research will focus on constellation shaping to optimize sidelobe performance under \ac{DPU}.

\clearpage
\balance 
\bibliographystyle{Sub/IEEEbib}
\bibliography{IEEEabrv,Sub/isac}
\end{document}